\documentclass[11pt]{article}

\usepackage[margin=0.72in]{geometry}
\usepackage{amsmath,amssymb,amsthm}
\usepackage{mathtools}
\usepackage{bm}
\usepackage{microtype}
\usepackage{multicol}
\usepackage{booktabs}
\usepackage{array}
\usepackage{cite}
\usepackage{hyperref}
\usepackage{xcolor}

\hypersetup{
    colorlinks=true,
    linkcolor=blue,
    citecolor=blue,
    urlcolor=blue
}

\newtheorem{theorem}{Theorem}[section]
\newtheorem{proposition}[theorem]{Proposition}
\newtheorem{corollary}[theorem]{Corollary}

\theoremstyle{definition}

\theoremstyle{remark}

\newcommand{\dd}{\mathrm{d}}
\newcommand{\Disc}{\operatorname{Disc}}
\newcommand{\cA}{\mathcal{A}}

\title{\bfseries Threshold Tails of Black-Hole Differential Observables}
\author{Vedant Subhash\\[2mm]
\small Department of Mathematics, University at Buffalo,\\
\small The State University of New York, Buffalo, New York 14260, USA}
\date{}

\begin{document}
\maketitle

\begin{abstract}
We study how first-order differential observables affect the zero-frequency behavior of black-hole wave equations. Quasinormal modes and late-time tails do not always change in the same way. We give a simple condition for when an observable removes the leading threshold term of a Green function. For Schwarzschild Regge-Wheeler modes, the relevant operator is determined by the regular static solution. The transformation is nondegenerate at nonzero frequency but becomes globally degenerate at zero frequency. As a result, the nonzero quasinormal-mode problem is unchanged, while the leading fixed-radius late-time tail gains one extra inverse power of time.
\end{abstract}

\vspace{5mm}

\begin{multicols}{2}
\raggedcolumns

\section{Introduction}

Black-hole perturbations contain different types of spectral information. Quasinormal modes appear as poles of a frequency-domain Green function. Late-time tails have a different origin. They come from nonanalytic behavior near zero frequency, usually through a branch cut. These two parts of the spectrum are related, but they are not the same.

The basic equations for Schwarzschild perturbations have been known for a long time. Regge and Wheeler derived the odd-parity gravitational equation~\cite{ReggeWheeler}, while Zerilli obtained the even-parity equation~\cite{Zerilli}. The relations between different forms of the Schwarzschild perturbation equations were studied in detail by Chandrasekhar~\cite{Chandrasekhar1975,ChandrasekharBook}. The wider framework of black-hole perturbation theory also includes the Teukolsky equation for rotating black holes~\cite{Teukolsky1973}.

The late-time behavior of black-hole perturbations was first studied in detail by Price~\cite{Price1972}. He showed that fields outside a Schwarzschild black hole decay with inverse powers of time. These decaying terms are now usually called Price tails. Later work showed that the decay rate depends on the field, the multipole number, the observation region, and the initial data~\cite{Andersson1997,BarackOri1999,Barack2000,PriceBurko}.Rigorous late-time asymptotics for scalar waves on spherically
symmetric stationary spacetimes were established in
Ref.~\cite{Angelopoulos2018}.

There is also a clear frequency-domain picture. Leaver showed that the black-hole response can be separated into different pieces, including quasinormal-mode poles and a branch-cut contribution~\cite{Leaver1986}. The relation between long-range potentials, branch points, and power-law tails was studied more generally by Ching, Leung, Suen, and Young~\cite{Ching1995}. Reviews of black-hole quasinormal modes and their physical meaning can be found in Refs.~\cite{KokkotasSchmidt,Berti2009}.

The Schwarzschild branch cut has been studied in much more detail in the last two decades. Casals and Ottewill developed analytic methods for the Green function on the branch cut and studied both low- and high-frequency behavior~\cite{CasalsOttewillPRL,CasalsOttewillLarge,CasalsOttewillBranch}. Related work also connected the branch cut with Green-function and self-force calculations~\cite{CasalsDolan}. Their later small-frequency analysis gave several terms in the late-time expansion for scalar, electromagnetic, and gravitational fields~\cite{CasalsOttewill}. In particular, it showed that logarithmic terms appear beyond the leading Price tail.

The Mano-Suzuki-Takasugi method is important for these calculations. It gives analytic series solutions of the Teukolsky and Regge-Wheeler equations and is well suited to small-frequency expansions~\cite{MSTTeukolsky,MSTRW}. These expansions will be used below as an input. We do not derive the Schwarzschild branch structure again.

There is also a mathematical theory of threshold behavior for Schr\"odinger operators. Resolvent expansions near zero energy were studied, for example, by Jensen and Kato and by Jensen and Nenciu~\cite{JensenKato,JensenNenciu}. These results show why zero frequency needs separate care. A resolvent near a threshold can contain powers, logarithms, or other nonanalytic terms. It cannot always be treated with the same meromorphic arguments that are used for isolated resonance poles.

A second part of the present problem concerns differential transformations. Darboux and supersymmetric transformations are standard tools for one-dimensional wave equations. Their effect on Jost functions was studied in early work on supersymmetric quantum mechanics~\cite{Talukdar1989}. Leung and collaborators developed the same ideas for open systems and quasinormal modes~\cite{Leung2001}. In black-hole perturbation theory, Glampedakis, Johnson, and Kennefick showed clearly how Darboux transformations appear in the Regge-Wheeler and Zerilli problems~\cite{Glampedakis2017}. These works also show that special frequencies, including zero-energy cases, require care.

The zero-frequency Schwarzschild equation has received renewed attention because of black-hole Love numbers. Static Schwarzschild black holes in four-dimensional general relativity have vanishing Love numbers. This result has been studied using direct response calculations, hidden symmetries, and ladder operators~\cite{Hui2021,Charalambous2021,Hui2022,BenAchour2022}. Recent work continues to develop this point of view and to clarify the role of ladder symmetries and static polynomial solutions~\cite{Sharma2026,DeLuca2026,Kumar2026}. These papers are important for the present work because the regular zero-frequency Regge--Wheeler solution is also the state that controls the transformation studied below.

The frequency-domain Schwarzschild Green function is also an active subject. Recent work has given new decompositions of the Green function and has studied its singular structure in more detail~\cite{Su2026,Rosato2026}. Aruquipa and Casals have constructed Regge-Wheeler and Teukolsky Green functions in Schwarzschild spacetime~\cite{Aruquipa2026}. Another recent study has shown how the frequency content of the source can strongly reduce a late-time tail~\cite{LeonVega2026}. These results make it especially important to separate an effect caused by the observable from an effect caused by the source.

The question studied here is the following. Let a black-hole master field $y$ satisfy a second-order radial equation. Suppose a physical quantity is obtained through a first-order differential map
\begin{equation}
    T(\omega)=A(x,\omega)+B(x,\omega)\partial_x .
    \label{eq:introT}
\end{equation}
How does this map change the zero-frequency branch contribution?

A related nonzero-frequency problem was studied in Ref.~\cite{SubhashPolePaper}. There the main object was the first-jet transformation determinant. An isolated zero of this determinant at one radial point can make the scalar equation for the transformed quantity singular. However, such a local zero does not by itself create a new quasinormal-mode pole. A genuine loss of a homogeneous solution occurs only when the transformation becomes globally degenerate. That analysis was carried out away from thresholds and branch points.

The present paper studies the missing zero-frequency case.

Our first result is general. Suppose the leading branch-cut part of the Green function has finite-rank range. Then the leading branch term in the transformed physical response disappears exactly when the limiting observable annihilates that threshold range. In the rank-one case this condition is simply
\begin{equation}
    T_0\phi_0=0,
    \label{eq:introTphi}
\end{equation}
where $\phi_0$ is the leading threshold state.

This condition also has a direct meaning for the first-jet transformation. If a regular first-order map kills a nonzero homogeneous threshold solution, then its first-jet determinant vanishes everywhere at that frequency. Thus suppression of the leading threshold state requires a global transformation degeneracy, not an isolated radial zero.

We then apply the result to the Schwarzschild Regge-Wheeler family. The horizon-regular zero-frequency solution is a finite polynomial. For each separated mode, all regular first-order maps that remove the leading threshold state have the form
\begin{equation}
    T_0=B(r)\left[\partial_{r_*}-\partial_{r_*}\ln\phi_{j\ell}(r)\right].
    \label{eq:introClass}
\end{equation}
For the canonical choice $B=1$, the transformation determinant is exactly
\begin{equation}
    D=\omega^2.
    \label{eq:introD}
\end{equation}

This gives a simple separation between the two spectral problems. At every nonzero frequency the canonical transformation is nondegenerate, so the nonzero quasinormal-mode problem is transported in the usual way. At $\omega=0$ the transformation becomes globally degenerate and removes the leading threshold state.

Using the known small-frequency Regge-Wheeler expansion, we then show that this zero is simple. The usual fixed-radius Price tail $t^{-2\ell-3}$ is therefore changed, for the one-sided transformed Green response, to $t^{-2\ell-4}$. The same change holds for a sourced solution when the leading threshold source overlap is nonzero.

The analysis uses several known ingredients, including Price tails,
the Schwarzschild branch cut, static polynomial solutions, and
zero-energy Darboux transformations. The point developed here is
the connection between them. For a general first-order response,
suppression of the leading threshold coefficient is tied to global
first-jet degeneracy. The Schwarzschild Regge-Wheeler problem gives
an exact example. The same map can preserve the nonzero-frequency
quasinormal-mode problem while changing the leading zero-frequency
tail. Thus two differential observables can have the same nonzero
quasinormal frequencies and still have different late-time tails.

The paper is organized as follows.
Section~\ref{sec:threshold} introduces the threshold Green-function
problem. Section~\ref{sec:degeneracy} relates threshold suppression
to global first-jet degeneracy. Section~\ref{sec:rw} studies the
zero-frequency Schwarzschild solutions.
Section~\ref{sec:classification} classifies the first-order
threshold suppressors, and Section~\ref{sec:determinant} gives
their exact determinant. Section~\ref{sec:tailshift} proves the
one-power change in the late-time tail.
Section~\ref{sec:data} discusses initial data and source
cancellation. Section~\ref{sec:darboux} compares the result with
Darboux theory, and Section~\ref{sec:checks} gives checks against
known Schwarzschild transformations. Section~\ref{sec:discussion}
discusses the interpretation and scope of the result. We end with
a short conclusion.

\section{Threshold response}
\label{sec:threshold}

We begin with a general second-order equation
\begin{equation}
    L(\omega)y=y''+p(x,\omega)y'+q(x,\omega)y,
    \label{eq:generalL}
\end{equation}
where $\omega$ is the frequency. The prime denotes differentiation with respect to $x$.

Let $R_L(\omega)$ denote the outgoing inverse with the physical boundary conditions. Near an ordinary resonance this inverse can be meromorphic. Near a threshold this need not be true. Powers, logarithms, and fractional powers can appear. This is standard in threshold scattering theory~\cite{JensenKato,JensenNenciu}.

We work on a fixed slit neighborhood of a threshold $\omega_0$. This fixes the branch of every logarithm and fractional power.

\subsection{Physical response and scalar transformed equation}

We study a first-order observable
\begin{equation}
    T(\omega)=A(x,\omega)+B(x,\omega)\partial_x .
    \label{eq:generalT}
\end{equation}
The physical transformed response is
\begin{equation}
    G_T(\omega)=T(\omega)R_L(\omega).
    \label{eq:physicalGT}
\end{equation}
This is not, in general, the same object as the unit-source Green function of a new scalar equation satisfied by $Ty$. This distinction is important both at poles and at thresholds.

If a source is present, the complete response is
\begin{equation}
    u(\omega)=T(\omega)R_L(\omega)J(\omega),
    \label{eq:completeResponse}
\end{equation}
where $J$ contains the physical source or the transformed initial data.

\subsection{Leading threshold coefficient}

Let $X$, $Y$, and $Z$ be Banach spaces. We regard the outgoing
resolvent and the observable as
\begin{equation}
    R_L(\omega):X\to Y,
    \qquad
    T(\omega):Y\to Z.
\end{equation}
We assume that these maps are bounded in a fixed slit neighborhood
of the threshold $\omega_0$.

Suppose that the jump of the resolvent has the leading expansion
\begin{equation}
    \Disc R_L(\omega)
    =
    f(\omega)P_0+E(\omega),
    \label{eq:leadingDisc}
\end{equation}
where $P_0\in\mathcal{B}(X,Y)$ is nonzero and
\begin{equation}
    \frac{\|E(\omega)\|_{\mathcal{B}(X,Y)}}{|f(\omega)|}
    \longrightarrow 0
    \qquad
    \text{as }\omega\to\omega_0
\end{equation}
along the cut. Here
$\Disc R_L=R_L^+-R_L^-$.

We also assume that $T$ is single-valued across the cut and regular
at the threshold in operator norm:
\begin{equation}
    T(\omega)
    =
    T_0+O(\omega-\omega_0),
    \label{eq:Tregular}
\end{equation}
with $T_0\in\mathcal{B}(Y,Z)$.

\begin{theorem}[Leading threshold map]
\label{thm:thresholdmap}
Under the assumptions above,
\begin{equation}
    \Disc\!\left[T(\omega)R_L(\omega)\right]=f(\omega)T_0P_0+o\!\left(f(\omega)\right).
    \label{eq:thresholdMap}
\end{equation}
Therefore the leading threshold contribution disappears if and only if
\begin{equation}
    T_0P_0=0.
    \label{eq:T0P0}
\end{equation}
\end{theorem}

\begin{proof}
Since $T$ is single-valued across the cut,
\begin{equation}
    \Disc(TR_L)=T\,\Disc R_L.
\end{equation}
Using Eqs.~\eqref{eq:leadingDisc} and \eqref{eq:Tregular},
\begin{align}
    \Disc(TR_L)
    &=
    \left[T_0+O(\omega-\omega_0)\right]
    \left[f(\omega)P_0+E(\omega)\right] \\
    &=
    f(\omega)T_0P_0+o\!\left(f(\omega)\right).
\end{align}
Therefore
\begin{equation}
    \Disc(TR_L)=o\!\left(f(\omega)\right)
\end{equation}
if and only if $T_0P_0=0$.
\end{proof}

The theorem is conditional on the existence of the stated threshold
expansion. We do not prove such an expansion for a general operator
$L$. In the Schwarzschild application below we instead use the known
fixed-radius small-frequency expansion of the Regge-Wheeler Green
function. The same argument then applies directly to its kernel. This theorem is elementary once the threshold expansion is known. Its value is that it identifies the correct object that must vanish. The condition is an operator condition, not a zero at one radial point.

\subsection{Rank-one threshold state}

A particularly useful case is
\begin{equation}
    P_0=\phi_0\otimes\lambda_0,
    \label{eq:rankone}
\end{equation}
where $\phi_0$ is a homogeneous threshold solution and $\lambda_0$ is a nonzero source functional.

\begin{corollary}[Rank-one criterion]
\label{cor:rankone}
If Eq.~\eqref{eq:rankone} holds, then the leading threshold term is suppressed for every source if and only if
\begin{equation}
    T_0\phi_0=0.
    \label{eq:Tphi0}
\end{equation}
\end{corollary}

\begin{proof}
Equation~\eqref{eq:T0P0} becomes $T_0P_0=(T_0\phi_0)\otimes\lambda_0$. Since $\lambda_0\neq0$, this vanishes exactly when $T_0\phi_0=0$.
\end{proof}

A special source may still remove the leading tail even when $T_0\phi_0\neq0$. That is a source cancellation. It should not be confused with a cancellation produced by the observable for every source.

\subsection{General polyhomogeneous expansion}

For completeness, suppose
\begin{equation}
    \Disc R_L\sim\sum_j\zeta^{\alpha_j}P_j(\ln\zeta),\qquad \zeta=\omega-\omega_0,
    \label{eq:polyR}
\end{equation}
and
\begin{align}
    T&\sim\sum_a\zeta^{\beta_a}T_a,\\
    J&\sim\sum_b\zeta^{\delta_b}J_b.
\end{align}
The first surviving exponent of the full response is
\begin{equation}
\begin{split}
    \gamma_*=\min\{&\alpha_j+\beta_a+\delta_b:\\
    &\sum_{\alpha_j+\beta_a+\delta_b=\gamma_*}T_aP_jJ_b\neq0\}.
    \label{eq:generalGamma}
\end{split}
\end{equation}
This form is useful when fractional powers are present. For integer-spaced Schwarzschild expansions it reduces to the simpler sum of the original threshold order, the observable order, the source order, and any remaining cancellation order.

If
\begin{equation}
    \Disc u(-i\sigma)\sim\sigma^\gamma(\ln\sigma)^k,\qquad \gamma>-1,
\end{equation}
then Laplace inversion gives
\begin{equation}
    u(t)\sim t^{-\gamma-1}Q_k(\ln t),
    \label{eq:laplaceTail}
\end{equation}
where $Q_k$ is a polynomial of degree at most $k$. This follows from
\begin{equation}
    \int_0^\infty e^{-\sigma t}\sigma^\gamma(\ln\sigma)^k\,\dd\sigma
    =\frac{\partial^k}{\partial\gamma^k}\left[\Gamma(\gamma+1)t^{-\gamma-1}\right].
    \label{eq:laplaceIdentity}
\end{equation}

\section{First-jet degeneracy}
\label{sec:degeneracy}

We now connect the threshold condition with the transformation determinant. For
\begin{equation}
    z=Ay+By',
\end{equation}
use the original equation to write
\begin{equation}
    z'=Cy+Ey',
\end{equation}
where
\begin{align}
    C&=A'-Bq,\\
    E&=A+B'-Bp.
\end{align}
The first jets satisfy
\begin{equation}
    \begin{pmatrix}z\\z'\end{pmatrix}
    =\begin{pmatrix}A&B\\C&E\end{pmatrix}
    \begin{pmatrix}y\\y'\end{pmatrix}.
\end{equation}
The determinant is
\begin{equation}
    D(x,\omega)=AE-BC,
    \label{eq:Dgeneral}
\end{equation}
or
\begin{equation}
    D=A^2+AB'-ABp-BA'+B^2q.
    \label{eq:Dexpanded}
\end{equation}
For two homogeneous solutions $y_1,y_2$,
\begin{equation}
    W[Ty_1,Ty_2]=D\,W[y_1,y_2].
    \label{eq:WronskianD}
\end{equation}

\begin{theorem}[Threshold suppression implies global degeneracy]
\label{thm:globaldeg}
Let the original second-order equation be regular on a connected interval. Let $\phi_0$ be a nonzero homogeneous solution at the threshold. If
\begin{equation}
    T_0\phi_0=0,
\end{equation}
then
\begin{equation}
    D(x,\omega_0)\equiv0
\end{equation}
on that interval.
\end{theorem}

\begin{proof}
Choose a second homogeneous solution $\psi_0$ that is independent of $\phi_0$. Then $W[\phi_0,\psi_0]\neq0$ on the connected regular interval. Using Eq.~\eqref{eq:WronskianD},
\begin{equation}
    D\,W[\phi_0,\psi_0]=W[T_0\phi_0,T_0\psi_0]=0.
\end{equation}
Hence $D\equiv0$.
\end{proof}

The converse needs care. If $D(\cdot,\omega_0)\equiv0$, then some nonzero homogeneous solution lies in the kernel of $T_0$. That solution does not have to be the particular threshold state that appears in $P_0$. Thus global degeneracy is necessary for universal suppression of a rank-one leading threshold state, but it is not sufficient unless the killed solution is the correct one.

\section{Schwarzschild Regge-Wheeler modes}
\label{sec:rw}

We now apply the general result to Schwarzschild spacetime. Let
\begin{equation}
    F(r)=1-\frac{2M}{r},
\end{equation}
and define the tortoise coordinate by
\begin{equation}
    \frac{\dd r_*}{\dd r}=\frac{1}{F}.
\end{equation}
For spin magnitude $j=0,1,2$ and multipole $\ell\ge j$, the generalized Regge-Wheeler equation is
\begin{equation}
    \frac{\dd^2 f}{\dd r_*^2}+\left[\omega^2-V_{j\ell}(r)\right]f=0,
    \label{eq:RW}
\end{equation}
with
\begin{equation}
    V_{j\ell}=F\left[\frac{\ell(\ell+1)}{r^2}+\frac{2M(1-j^2)}{r^3}\right].
    \label{eq:RWpotential}
\end{equation}
The cases $j=0,1,2$ correspond to the scalar, electromagnetic, and axial gravitational sectors.

\subsection{The static solution}

Set $x=r/(2M)$. At $\omega=0$, Eq.~\eqref{eq:RW} becomes
\begin{equation}
    x^2(x-1)\phi''+x\phi'-\left[\ell(\ell+1)x+1-j^2\right]\phi=0,
    \label{eq:staticX}
\end{equation}
where the primes in this subsection denote $x$ derivatives.

Put $\phi=x^{j+1}u$. Then $u$ satisfies
\begin{equation}
\begin{split}
    x(1-x)u''+\left[2j+1-(2j+2)x\right]u'\\
    -(j-\ell)(j+\ell+1)u=0.
    \label{eq:hypergeom}
\end{split}
\end{equation}
Let $n=\ell-j$. The horizon-regular solution normalized by $\phi_{j\ell}(2M)=1$ is
\begin{equation}
\begin{split}
    \phi_{j\ell}(x)={}&(-1)^n\frac{(\ell+j)!}{n!(2j)!}x^{j+1}\\
    &\times{}_2F_1\!\left(-n,\ell+j+1;2j+1;x\right).
    \label{eq:staticPolynomial}
\end{split}
\end{equation}
Since $-n$ is a nonpositive integer, the hypergeometric series terminates. Thus $\phi_{j\ell}$ is a polynomial of degree $\ell+1$. This polynomial truncation is known and is closely related to the static Love-number structure of Schwarzschild black holes~\cite{Hui2021,Hui2022,Kumar2026}.

\subsection{No exterior zero}

\begin{proposition}
\label{prop:positive}
The horizon-regular static solution satisfies
\begin{equation}
    \phi_{j\ell}(r)>0
\end{equation}
for every $r>2M$.
\end{proposition}

\begin{proof}
For $\ell\ge j$,
\begin{equation}
    \ell(\ell+1)+1-j^2\ge j+1>0.
\end{equation}
Hence
\begin{equation}
    V_{j\ell}(r)>0
\end{equation}
for every $r>2M$. At zero frequency the equation is
\begin{equation}
    \frac{\dd^2\phi}{\dd r_*^2}
    =
    V_{j\ell}\phi.
    \label{eq:staticConvex}
\end{equation}

The horizon-regular normalization gives
\begin{equation}
    \phi\to1,
    \qquad
    \frac{\dd\phi}{\dd r_*}\to0
    \qquad
    \text{as }r_*\to-\infty.
\end{equation}

Suppose, for contradiction, that $\phi$ has a first zero at a finite
point $r_{*,0}$. Then $\phi>0$ for every $r_*<r_{*,0}$. On this
interval Eq.~\eqref{eq:staticConvex} gives
\begin{equation}
    \phi''=V_{j\ell}\phi>0.
\end{equation}
Since the potential tends to zero exponentially at the horizon, the
horizon conditions allow us to integrate from $-\infty$:
\begin{equation}
    \phi'(r_*)
    =
    \int_{-\infty}^{r_*}
    V_{j\ell}(s)\phi(s)\,\dd s.
\end{equation}
The integrand is positive before the assumed first zero. Therefore
\begin{equation}
    \phi'(r_*)>0
\end{equation}
for every finite $r_*<r_{*,0}$. It follows that $\phi$ is strictly
increasing from its horizon value $1$. Hence it cannot reach zero at
$r_{*,0}$. This is a contradiction.

Therefore
\begin{equation}
    \phi_{j\ell}(r)>0
\end{equation}
for every $r>2M$.
\end{proof}

Therefore $\partial_{r_*}\ln\phi_{j\ell}$ is regular over the whole exterior region. For the lowest physical multipoles,
\begin{align}
    \phi_{00}&=x,\\
    \phi_{11}&=x^2,\\
    \phi_{22}&=x^3.
    \label{eq:lowestPhi}
\end{align}

\section{Classification of threshold suppressors}
\label{sec:classification}

For the Schwarzschild problem, the branch-cut discontinuity of the
radial Green function has a factorized form in terms of the ingoing
solution~\cite{CasalsOttewill,CasalsOttewillBranch}. At fixed exterior
radii we write
\begin{equation}
    \Disc G_{j\ell}(r,r';-i\sigma)
    =
    \mathcal{B}_{j\ell}(\sigma)
    f_{j\ell}(r,-i\sigma)
    f_{j\ell}(r',-i\sigma),
    \label{eq:BCfactorized}
\end{equation}
where the scalar branch-cut factor contains the Wronskian and
branch-cut strength. Its small-frequency behavior is
\begin{equation}
    \mathcal{B}_{j\ell}(\sigma)
    =
    C_{j\ell}\sigma^{2\ell+2}
    +O(\sigma^{2\ell+3}),
    \qquad
    C_{j\ell}\neq0.
    \label{eq:BCfactor}
\end{equation}
Since
\begin{equation}
    f_{j\ell}(r,-i\sigma)
    =
    \phi_{j\ell}(r)+O(\sigma),
\end{equation}
the leading fixed-radius discontinuity is therefore
\begin{equation}
\begin{split}
    \Disc G_{j\ell}
    ={}&
    C_{j\ell}\sigma^{2\ell+2}
    \phi_{j\ell}(r)\phi_{j\ell}(r') \\
    &+O(\sigma^{2\ell+3}).
    \label{eq:RWleadingDisc}
\end{split}
\end{equation}
Thus the leading threshold coefficient is rank one.

\begin{theorem}[First-order threshold suppressors]
\label{thm:classification}
Let
\begin{equation}
    T_0=A(r)+B(r)\partial_{r_*}
    \label{eq:T0RW}
\end{equation}
be regular outside the Schwarzschild horizon, with $B$ not identically zero. The leading fixed-radius threshold term in Eq.~\eqref{eq:RWleadingDisc} is removed for every source point if and only if
\begin{equation}
    T_0=B(r)\left[\partial_{r_*}-\partial_{r_*}\ln\phi_{j\ell}(r)\right].
    \label{eq:classifiedT}
\end{equation}
\end{theorem}

\begin{proof}
By Corollary~\ref{cor:rankone}, universal suppression is equivalent to $T_0\phi_{j\ell}=0$. Hence
\begin{equation}
    A\phi_{j\ell}+B\phi_{j\ell}'=0.
\end{equation}
Since $\phi_{j\ell}$ has no exterior zero,
\begin{equation}
    A=-B\frac{\phi_{j\ell}'}{\phi_{j\ell}},
\end{equation}
which gives Eq.~\eqref{eq:classifiedT}.
\end{proof}

Thus the canonical suppressor is
\begin{equation}
    \boxed{\cA_{j\ell}=\partial_{r_*}-\partial_{r_*}\ln\phi_{j\ell}}.
    \label{eq:canonicalA}
\end{equation}
For the lowest scalar, electromagnetic, and gravitational modes,
\begin{align}
    \cA_{00}&=\partial_{r_*}-\frac{F}{r},\\
    \cA_{11}&=\partial_{r_*}-\frac{2F}{r},\\
    \cA_{22}&=\partial_{r_*}-\frac{3F}{r}.
    \label{eq:lowestA}
\end{align}
These operators are mode dependent. They should therefore be understood as transformations of separated radial master equations, not as one universal local four-dimensional field operator.

\section{Exact threshold degeneracy}
\label{sec:determinant}

Consider
\begin{equation}
    T=B(r)(\partial_{r_*}-w),\qquad w=\frac{\phi_{j\ell}'}{\phi_{j\ell}},
    \label{eq:generalSuppressor}
\end{equation}
where the prime now denotes an $r_*$ derivative. For the Regge-Wheeler equation, $p=0$ and $q=\omega^2-V$. Since $\phi_{j\ell}''=V\phi_{j\ell}$,
\begin{equation}
    w'+w^2=V.
    \label{eq:Riccati}
\end{equation}

\begin{theorem}[Exact determinant]
\label{thm:Domega}
For the threshold suppressor Eq.~\eqref{eq:generalSuppressor}, the first-jet determinant is
\begin{equation}
    \boxed{D(r,\omega)=B(r)^2\omega^2}.
    \label{eq:Domega}
\end{equation}
For the canonical choice $B=1$, $D=\omega^2$.
\end{theorem}

\begin{proof}
Use $A=-Bw$ in Eq.~\eqref{eq:Dexpanded}. The terms containing $B'$ cancel. The remaining result is
\begin{equation}
    D=B^2(w^2+w'+\omega^2-V).
\end{equation}
Equation~\eqref{eq:Riccati} gives Eq.~\eqref{eq:Domega}.
\end{proof}

For the canonical suppressor, $D\neq0$ for every $\omega\neq0$, while $D(\cdot,0)\equiv0$. Thus the canonical transformation becomes globally degenerate exactly at the frequency where it removes the leading threshold state.

At nonzero frequency, $w\to0$ at the horizon and at infinity. Therefore
\begin{align}
    \cA_{j\ell}e^{-i\omega r_*}&\sim-i\omega e^{-i\omega r_*},\\
    \cA_{j\ell}e^{+i\omega r_*}&\sim+i\omega e^{+i\omega r_*}.
\end{align}
The physical ingoing and outgoing boundary lines are therefore preserved for nonzero $\omega$. The nonzero-frequency quasinormal-mode problem is transported, while $\omega=0$ remains exceptional.

\section{The late-time tail shifts by one power}
\label{sec:tailshift}

We now show that the canonical suppressor removes exactly one threshold order. Write $\epsilon=2M\sigma$ and expand the ingoing solution as
\begin{equation}
    f_{j\ell}(r,-i\sigma)=\phi_{j\ell}(r)+\epsilon f_{1,j\ell}(r)+O(\epsilon^2).
    \label{eq:fExpansion}
\end{equation}
Such small-frequency expansions follow naturally from the MST construction~\cite{MSTRW,CasalsOttewill}.

Because the radial equation depends on frequency through $\omega^2$,
\begin{align}
    L_0\phi_{j\ell}&=0,\\
    L_0f_{1,j\ell}&=0.
\end{align}
The ingoing normalization at the horizon is $f_{j\ell}\sim e^{-i\omega r_*}$. On the negative imaginary axis,
\begin{equation}
    e^{-i\omega r_*}=e^{-\sigma r_*}=1-\epsilon\frac{r_*}{2M}+O(\epsilon^2).
\end{equation}
Therefore
\begin{equation}
    W[\phi_{j\ell},f_{1,j\ell}]=-\frac{1}{2M}.
    \label{eq:staticWronskian}
\end{equation}
Now apply the canonical suppressor:
\begin{equation}
    \cA_{j\ell}f_{1,j\ell}=\frac{W[\phi_{j\ell},f_{1,j\ell}]}{\phi_{j\ell}}.
\end{equation}
Using Eq.~\eqref{eq:staticWronskian},
\begin{equation}
    \boxed{\cA_{j\ell}f_{1,j\ell}=-\frac{1}{2M\phi_{j\ell}}}.
    \label{eq:Aonf1}
\end{equation}
It is nonzero everywhere outside the horizon. Hence
\begin{equation}
    \cA_{j\ell}f_{j\ell}(r,-i\sigma)=-\frac{\sigma}{\phi_{j\ell}(r)}+O(\sigma^2).
    \label{eq:Afsmall}
\end{equation}
The zero is therefore simple.

\begin{theorem}[One-power tail shift]
\label{thm:tailshift}
At fixed radii $2M<r,r'<\infty$, the ordinary Regge-Wheeler branch Green kernel has
\begin{equation}
    G_{j\ell}^{\mathrm{BC}}(t;r,r')\sim t^{-2\ell-3}.
\end{equation}
For the one-sided canonical transformed kernel,
\begin{equation}
    \cA_{j\ell,r}G_{j\ell}^{\mathrm{BC}}(t;r,r')\sim t^{-2\ell-4}.
    \label{eq:tailResult}
\end{equation}
If this kernel is integrated against a physical source or initial-data term, the same leading power holds provided the corresponding leading threshold overlap is nonzero.
\end{theorem}

\begin{proof}
Using the factorized branch-cut expression
\eqref{eq:BCfactorized},
\begin{equation}
\begin{split}
    \Disc\!\left(
    \cA_{j\ell,r}G_{j\ell}
    \right)
    ={}&
    \mathcal{B}_{j\ell}(\sigma)
    \left[
    \cA_{j\ell}
    f_{j\ell}(r,-i\sigma)
    \right] \\
    &\times
    f_{j\ell}(r',-i\sigma).
\end{split}
\end{equation}
Equations~\eqref{eq:BCfactor}, \eqref{eq:fExpansion}, and
\eqref{eq:Afsmall} give
\begin{equation}
\begin{split}
    \Disc\!\left(
    \cA_{j\ell,r}G_{j\ell}
    \right)
    ={}&
    \left[
    C_{j\ell}\sigma^{2\ell+2}
    +O(\sigma^{2\ell+3})
    \right] \\
    &\times
    \left[
    -\frac{\sigma}{\phi_{j\ell}(r)}
    +O(\sigma^2)
    \right]
    \left[
    \phi_{j\ell}(r')+O(\sigma)
    \right].
\end{split}
\end{equation}
Therefore
\begin{equation}
    \boxed{
    \Disc\!\left(
    \cA_{j\ell,r}G_{j\ell}
    \right)
    =
    -C_{j\ell}
    \frac{\phi_{j\ell}(r')}
         {\phi_{j\ell}(r)}
    \sigma^{2\ell+3}
    +O(\sigma^{2\ell+4})
    }.
    \label{eq:transformedBCcoefficient}
\end{equation}
The coefficient is nonzero because
$C_{j\ell}\neq0$ and
$\phi_{j\ell}(r),\phi_{j\ell}(r')>0$ for finite exterior radii.

Thus the first surviving branch term is exactly of order
$\sigma^{2\ell+3}$, not merely at least of that order.
Laplace inversion then gives
\begin{equation}
    \cA_{j\ell,r}
    G_{j\ell}^{\mathrm{BC}}(t;r,r')
    \sim t^{-2\ell-4}.
\end{equation}
For a sourced solution, this coefficient survives provided the
leading source overlap is nonzero.
\end{proof}

For the lowest physical modes,
\begin{align}
    j=\ell=0:&\qquad t^{-3}\longrightarrow t^{-4},\\
    j=\ell=1:&\qquad t^{-5}\longrightarrow t^{-6},\\
    j=\ell=2:&\qquad t^{-7}\longrightarrow t^{-8}.
    \label{eq:lowestTailShifts}
\end{align}
The standard Price-law powers on the left-hand side are known results. The statement here concerns their transformation under the one-sided threshold suppressor.

\section{Initial data and source cancellation}
\label{sec:data}

The source must be treated separately from the observable. Consider
\begin{equation}
    \partial_t^2\Psi-\partial_{r_*}^2\Psi+V\Psi=0.
    \label{eq:timeRW}
\end{equation}
Let
\begin{align}
    \Psi(0,r_*)&=\Psi_0(r_*),\\
    \partial_t\Psi(0,r_*)&=\Pi_0(r_*).
\end{align}
With the retarded transform
\begin{equation}
    \widehat\Psi(\omega,r_*)=\int_0^\infty e^{i\omega t}\Psi(t,r_*)\,\dd t,
\end{equation}
integration by parts gives
\begin{equation}
    [\partial_{r_*}^2+\omega^2-V]\widehat\Psi=i\omega\Psi_0-\Pi_0.
    \label{eq:initialSource}
\end{equation}
Thus
\begin{equation}
    J(\omega)=i\omega\Psi_0-\Pi_0.
    \label{eq:Jinitial}
\end{equation}

For generic compact data, the leading source term is usually nonzero at $\omega=0$. The standard fixed-radius tail is then $t^{-2\ell-3}$. For momentarily stationary compact data, $\Pi_0=0$, the source contains one extra factor of $\omega$. The known late-time behavior becomes $t^{-2\ell-4}$, in agreement with Price and Burko~\cite{PriceBurko}.

Applying the canonical suppressor gives one further threshold zero. For generic data with nonzero leading overlap,
\begin{equation}
    \cA_{j\ell}\Psi\sim t^{-2\ell-4}.
\end{equation}
For momentarily stationary data, provided the next overlap does not also vanish,
\begin{equation}
    \cA_{j\ell}\Psi\sim t^{-2\ell-5}.
    \label{eq:staticDataShift}
\end{equation}
There are also nonstationary data for which the leading threshold overlap vanishes. Such a cancellation belongs to the source, not to the observable. Recent work has emphasized how the spectral content of the source can strongly change tail amplitudes~\cite{LeonVega2026}.

\section{Relation to Darboux theory}
\label{sec:darboux}

The canonical suppressor is also a zero-energy Darboux operator. This connection is classical. Define
\begin{equation}
    H=-\partial_{r_*}^2+V,
\end{equation}
and $w=\phi'/\phi$. Because $w'+w^2=V$,
\begin{equation}
    H=\cA^\dagger\cA,
\end{equation}
where
\begin{align}
    \cA&=\partial_{r_*}-w,\\
    \cA^\dagger&=-\partial_{r_*}-w.
\end{align}
The partner operator is $\widetilde H=\cA\cA^\dagger$. Zero-energy and open-system supersymmetric transformations have been studied before~\cite{Leung2001,Glampedakis2017}. We do not claim this factorization as new.

There is, however, an important distinction between $\cA R_H$ and the unit-source Green function of $\widetilde H$. Let $z=\omega^2$ and define
\begin{align}
    R(z)&=(H-z)^{-1},\\
    \widetilde R(z)&=(\widetilde H-z)^{-1}.
\end{align}
Intertwining gives $\cA R=\widetilde R\cA$. Multiplying on the right by $\cA^\dagger$ gives
\begin{equation}
    \cA R\cA^\dagger=I+z\widetilde R.
\end{equation}
Therefore
\begin{equation}
    \boxed{\widetilde R(z)=\frac{\cA R(z)\cA^\dagger-I}{z}}.
    \label{eq:partnerResolvent}
\end{equation}

This distinction is important for the main result of this paper.
The one-power tail shift proved in
Theorem~\ref{thm:tailshift} is a statement about the physical
one-sided response
\begin{equation}
    \cA R,
\end{equation}
not about the unit-source Green function $\widetilde R$ of the
Darboux partner problem. The two objects differ precisely by the
right action of $\cA^\dagger$ and by the factor $1/\omega^2$ in
Eq.~\eqref{eq:partnerResolvent}. Their threshold powers therefore
need not agree.

\section{Checks against known Schwarzschild transformations}
\label{sec:checks}

The general criterion gives several useful checks. First consider $T=\partial_{r_*}$. For the gravitational quadrupole,
\begin{equation}
    \phi_{22}=\left(\frac{r}{2M}\right)^3,
\end{equation}
so
\begin{equation}
    \partial_{r_*}\phi_{22}=\frac{3F}{r}\phi_{22},
\end{equation}
which is nonzero at every ordinary fixed exterior radius. The derivative does not suppress the leading threshold state. Hence the fixed-radius leading power remains $t^{-7}$.

The Regge-Wheeler and Zerilli equations are related by a Chandrasekhar transformation. This map is regular at zero frequency for radiative multipoles and does not generically kill the leading static state. The two parity sectors therefore provide a regular non-annihilating check of the threshold criterion. Recent frequency-domain work treats these parity relations at the level of the Green function~\cite{Rosato2026}.

The transformation between Regge-Wheeler and curvature variables is also well studied. Recent work has constructed the Regge-Wheeler and Teukolsky Green functions directly in Schwarzschild spacetime~\cite{Aruquipa2026}. These results should be viewed as checks and context for the present framework, not as new consequences claimed here.

\section{Discussion}
\label{sec:discussion}

The results above separate two parts of the black-hole spectral
problem that behave differently under a first-order differential
observable. At nonzero frequency, the relevant transformation can
remain invertible and preserve the physical ingoing and outgoing
boundary conditions. At the zero-frequency threshold, however, the
same transformation can become globally degenerate. It can then
remove the leading threshold state even though the nonzero-frequency
quasinormal-mode problem is unchanged.

Several ingredients used in this analysis are already known. The
Schwarzschild late-time tail and its branch-cut origin are classical
results. The horizon-regular static solutions are also known, as is
their relation to the static response of Schwarzschild black holes.
The Darboux factorization associated with a zero-energy solution is
standard as well. The point of the present analysis is the connection
between these ingredients.

For a general first-order physical observable, the leading threshold
term is controlled by the action of the limiting operator on the
range of the leading branch coefficient. In the rank-one case this
reduces to
\begin{equation}
    T_0\phi_0=0.
\end{equation}
This condition is stronger than the appearance of an isolated zero
of the first-jet determinant. It forces the transformation to become
globally degenerate at the threshold.

The Schwarzschild ReggeWheeler problem gives an exact example. The
regular static solution determines all regular first-order operators
that suppress the leading threshold state. For the canonical choice,
\begin{equation}
    \cA_{j\ell}
    =
    \partial_{r_*}
    -
    \partial_{r_*}\ln\phi_{j\ell},
\end{equation}
the first-jet determinant is
\begin{equation}
    D=\omega^2.
\end{equation}
Thus the map is nondegenerate for every nonzero frequency and
degenerate exactly at the threshold.

This gives a direct connection with the earlier nonzero-frequency
analysis in Ref.~\cite{SubhashPolePaper}. Away from the threshold,
local zeros of the transformation determinant do not by themselves
produce new quasinormal-mode poles. The present result shows that the
zero-frequency case is different. Suppression of the leading
threshold state occurs when the transformation loses one homogeneous
solution globally.

The effect can be seen directly in the branch-cut coefficient. At
fixed exterior radii,
\begin{equation}
    \Disc\!\left(
    \cA_{j\ell,r}G_{j\ell}
    \right)
    =
    -C_{j\ell}
    \frac{\phi_{j\ell}(r')}
         {\phi_{j\ell}(r)}
    \sigma^{2\ell+3}
    +O(\sigma^{2\ell+4}),
\end{equation}
so the leading transformed coefficient is nonzero. The fixed-radius
tail therefore changes from
\begin{equation}
    t^{-2\ell-3}
    \quad\hbox{to}\quad
    t^{-2\ell-4}.
\end{equation}

The result concerns the one-sided physical response
$\cA R$. It should not be identified with the unit-source Green
function of the Darboux partner equation. At a zero-energy
factorization point the partner resolvent contains an additional
factor of $1/\omega^2$, so its threshold behavior need not have the
same power.

There are also some natural boundaries to the present analysis. We
use the known small-frequency Schwarzschild expansion rather than
constructing a new weighted-space resolvent theory from first
principles. The canonical suppressor is mode dependent and therefore
acts on separated radial equations rather than as one universal
local spacetime operator. Our tail result is for fixed exterior
radius; the event horizon and future null infinity require separate
endpoint analyses. Singular observables containing factors such as
$1/\omega$ also require a separate treatment because they may change
the regular zero-frequency part of the response.

These points do not affect the fixed-radius result proved here. They
instead indicate where the threshold formulation can be extended.
In particular, it would be useful to study the same mechanism at the
horizon and null infinity, for higher-order differential observables,
and for rotating black-hole equations where the zero-frequency
structure is more complicated.

\section{Conclusion}

We studied how a first-order differential observable changes the
zero-frequency threshold response of a black-hole master equation.

For a finite-rank leading threshold term, the transformed coefficient
is obtained by applying the limiting observable to the threshold
range. In the rank-one case, suppression of the leading term requires
\begin{equation}
    T_0\phi_0=0.
\end{equation}
For a first-order map this condition forces the first-jet determinant
to vanish globally at the threshold.

For Schwarzschild Regge-Wheeler modes, the horizon-regular static
solution determines the complete family of regular first-order
threshold suppressors. For the canonical operator,
\begin{equation}
    \cA_{j\ell}
    =
    \partial_{r_*}
    -
    \partial_{r_*}\ln\phi_{j\ell},
\end{equation}
the determinant is exactly
\begin{equation}
    D=\omega^2.
\end{equation}
The transformation is therefore nondegenerate at every nonzero
frequency but globally degenerate at zero frequency.

The transformed branch-cut coefficient begins one power of frequency
later than the original one. Consequently, at fixed exterior radius,
\begin{equation}
    t^{-2\ell-3}
    \longrightarrow
    t^{-2\ell-4}.
\end{equation}

The main point is therefore simple. Two differential observables can
have the same nonzero quasinormal-mode spectrum and still have
different late-time tails. The difference is controlled by the
zero-frequency threshold, where the transformation can become
globally degenerate.

\appendix
\section{Static hypergeometric solution}

For reference, we give the short derivation of Eq.~\eqref{eq:staticPolynomial}. At zero frequency the radial equation in $x=r/(2M)$ is
\begin{equation}
    x^2(x-1)\phi''+x\phi'-[\ell(\ell+1)x+1-j^2]\phi=0.
\end{equation}
Set $\phi=x^{j+1}u$. A direct substitution gives
\begin{equation}
\begin{split}
    x(1-x)u''+[2j+1-(2j+2)x]u'\\
    -(j-\ell)(j+\ell+1)u=0.
\end{split}
\end{equation}
This is the hypergeometric equation with
\begin{align}
    a&=j-\ell,\\
    b&=j+\ell+1,\\
    c&=2j+1.
\end{align}
Because $a=-(\ell-j)$ is a nonpositive integer, the series terminates. At $x=1$,
\begin{equation}
    {}_2F_1(-n,b;c;1)=\frac{(c-b)_n}{(c)_n},
\end{equation}
which fixes the normalization used in Eq.~\eqref{eq:staticPolynomial}.

\section{First-jet determinant for the suppressor}

For a Schr\"odinger-type equation
\begin{equation}
    y''+(\omega^2-V)y=0
\end{equation}
and $T=A+B\partial_x$, the first-jet determinant is
\begin{equation}
    D=A^2+AB'-BA'+B^2(\omega^2-V).
\end{equation}
Take $A=-Bw$ with $w=\phi'/\phi$. Then
\begin{equation}
    A^2+AB'-BA'=B^2(w^2+w').
\end{equation}
The zero-frequency equation gives $w'+w^2=V$. Therefore
\begin{equation}
    D=B^2\omega^2.
\end{equation}

\section{Initial-value source}

Let
\begin{equation}
    \widehat\Psi(\omega,r_*)=\int_0^\infty e^{i\omega t}\Psi(t,r_*)\,\dd t.
\end{equation}
Integrating twice by parts gives
\begin{equation}
\begin{split}
    \int_0^\infty e^{i\omega t}\partial_t^2\Psi\,\dd t
    =-\Pi_0+i\omega\Psi_0-\omega^2\widehat\Psi.
\end{split}
\end{equation}
Hence
\begin{equation}
    [\partial_{r_*}^2+\omega^2-V]\widehat\Psi=i\omega\Psi_0-\Pi_0.
\end{equation}

\end{multicols}
\end{document}